\documentclass[paper=letter, fontsize=12pt]{article}
\usepackage[english]{babel} 
\usepackage{amsmath,amsfonts,amsthm} 
\usepackage[utf8]{inputenc}
\usepackage{float}
\usepackage{lipsum} 
\usepackage{blindtext}
\usepackage{graphicx}
\usepackage{color}
\usepackage{caption}
\usepackage{subcaption}
\usepackage[sc]{mathpazo} 
\usepackage[T1]{fontenc} 
\usepackage{microtype} 
\usepackage[hmarginratio=1:1,top=32mm,columnsep=20pt]{geometry} 
\usepackage{multicol} 
\usepackage{booktabs} 
\usepackage{float} 
\usepackage{hyperref} 
\usepackage{lettrine} 
\usepackage{paralist} 
\usepackage{abstract} 
\usepackage{titlesec} 
\usepackage{framed}

\renewcommand\thesection{\Roman{section}} 
\renewcommand\thesubsection{\Roman{subsection}} 

\titleformat{\section}[block]{\large\scshape\centering}{\thesection.}{1em}{} 
\titleformat{\subsection}[block]{\large}{\thesubsection.}{1em}{} 
\usepackage{fancyhdr} 
\usepackage{amsthm}

\title{\vspace{-15mm}\fontsize{20pt}{10pt}\selectfont\textbf{Non-Gaussian analytic option pricing: a closed formula for the L\'evy-stable model}} 
\author{
{\normalsize Jean-Philippe Aguilar$^{a,\dagger}$, Cyril Coste$^{b,\dagger\dagger}$, Jan Korbel$^{c,\dagger\dagger\dagger}$ }\\
{\small \it $\dagger$ BRED Banque Populaire, Modeling Department, 18 quai de la Râpée, Paris - 75012}\\
{\small \it $\dagger$ $\dagger$ MAIF, 200 avenue Salvador Allende, Niort}\\
{\small \it $\dagger$ $\dagger$ $\dagger$ Section for the Science of Complex Systems, CeMSIIS, Medical
University of Vienna}\\
{\small \it Spitalgasse 23, A-1090, Vienna, Austria}\\
{\small \it Complexity Science Hub Vienna, Josefst\"{a}dterstrasse 39, 1080 Vienna, Austria}\\
{\small \it Faculty of Nuclear Sciences and Physical Engineering, Czech Technical University in Prague}\\
{\small \it B\v{r}ehov\'{a} 7, 115 19, Prague, Czech Republic}\\
{\normalsize \textit{(a)} jean-philippe.aguilar@bred.fr, \textit{(b)} cyril.coste@maif.fr,}\\
{\normalsize (c) korbeja2@fjfi.cvut.cz}
}

\date{}

\providecommand{\keywords}[1]{\textbf{\textit{Key words---}} #1}

\newcommand{\res}{\mathrm{Res}}
\newcommand{\id}{\mathrm{d}}
\newcommand{\ud}{\mathrm{d}}

\newtheorem{lemma}{Lemma}
\newtheorem{proposition}{Proposition}
\newtheorem{theorem}{Theorem}

\begin{document}
\maketitle 
\pagestyle{headings}
\setcounter{page}{1}
\pagenumbering{arabic}
\begin{abstract}
\noindent We establish an explicit pricing formula for the class of L\'evy-stable models with maximal negative asymmetry (Log-L\'evy model with finite moments and stability parameter $1<\alpha\leq 2$) in the form of rapidly converging series. The series is obtained with help of Mellin transform and the residue theory in $\mathbb{C}^2$. The resulting formula enables the straightforward evaluation of an European option with arbitrary accuracy without the use of numerical techniques. The formula can be used by any practitioner, even if not familiar with the underlying mathematical techniques. We test the efficiency of the formula, and compare it with numerical methods.
\end{abstract}


\thispagestyle{fancy} 

\section{Introduction}

Black-Scholes (BS) model \cite{BS73} is the well established option pricing model which opened the way to the development of theory derivative pricing, and remains a fundamental tool in the construction of hedging policies. It is a Gaussian model --  the underlying price dynamics is described by a geometric Brownian motion. Nice feature of the BS model is that the solution is analytically solvable, i.e., the price of an European option under this model can be written out explicitly via a formula involving only elementary functions.
\newline
\noindent Black-Scholes theory, however, does not describe well extreme events, such as dramatic price drops, which appear  more often than predicted by Gaussian models \cite{Taleb10,Kleinert13}. As a result, many generalizations of BS models that incorporate extreme events have been introduced.
Particularly interesting are models based on $\alpha$-stable (or also called L\'{e}vy) distributions. These distributions possess polynomial decays of the tails which makes large jumps much more probable. The first signs of heavy tails in financial modeling have been discussed in seminal works of Mandelbrot and Fama in 1960s \cite{Mandelbrot63,Fama65}. The L\'{e}vy distributions decay as $1/|x|^{1+\alpha}$, where $\alpha\in  [0,2]$ is called \textit{stability parameter}. For financial applications, it is typical to assume $\alpha\in ]1,2]$, because the sample paths are a.s. continuous; for most financial applications, $\alpha \approx 1.5-1.7$ -- see \cite{Bouchaud97} and references therein. When $\alpha=2$, we recover the ordinary Gaussian distribution.
\newline
However, in general, price (or log-price) $S_t$ driven by a L\'evy distribution does not possess finite moments $E[S_t^n]$, resulting in potentially infinite option values (see \cite{Wu03} and references therein). To cope with this divergence, Carr and Wu introduced, 
a model featuring a L\'evy-driven underlying price \cite{Wu03}, with a so-called condition of \textit{maximal negative asymmetry} (or \textit{skewness}) on the Levy process that ensures the option price to remain finite. Following Carr and Wu, we will refer to this model as the \textit{FMLS} (Finite Moment Log Stable) or simply the \textit{L\'evy-stable} model.
\newline
Existing option pricing methods under L\'evy stable model include Monte-Carlo related simulations \cite{Tankov06}, which converge with a very good degree of precision (although structurally possessing a statistical error), and numerical evaluation of integrals (Mellin-Barnes integrals in \cite{Kleinert16}, integrals on the L\'evy density in \cite{Rachev99}). However, unlike to Black-Scholes model, no simple closed formula existed to evaluate the option price directly. The purpose of this paper is to establish such a formula for the L\'evy-stable model; it will be achieved under the form of the rapidly convergent series 
, whose terms are very simple and straightforward to compute. The calculation of the series is based on techniques known from complex analysis, particularly Mellin-Barnes integral representation and residue theory in $\mathbb{C}^2$.
\newline
\noindent The paper is organized as follows: in the next section we briefly recall some facts and notations about L\'evy-stable option pricing. In section 3, we establish a Mellin-Barnes integral representation for the call price in $\mathbb{C}^2$. In section 4, we compute this integral by means of residue summation, and prove the pricing formula \eqref{Formula}. In section 5, we demonstrate the efficiency of this formula and we check it by comparing to results obtained via numerical tools. For the readers' convenience, we also add an appendix with finer details of Mellin transform in $\mathbb{C}$ and $\mathbb{C}^2$, and some tools from fractional analysis that turn out to be useful for stable option pricing.

\section{L\'{e}vy stable option pricing}

\noindent In this section, we briefly recall the definition and main properties of $\alpha$-stable distributions and corresponding L\'{e}vy process. We focus on the application to financial modeling, particularly to Finite moment log-stable processes. 

\noindent The log-stable, (or log-L\'evy model), is a non-Gaussian model into which the underlying price is assumed to be described by the stochastic differential equation
\begin{equation}\label{Levy_SDE}
\ud S_t \, = \, r S_t \ud t \, +  \, \sigma \,  S_t \ud L_{\alpha,\beta}(t)
\end{equation}
where $L_{\alpha,\beta}(t)$ is the L\'evy process. The probability distribution of the L\'{e}vy process is the $\alpha$-stable distribution $g_{\alpha,\beta}(x,t) = \frac{1}{t^{1/\alpha}} g\left(\frac{x}{t^{1/\alpha}}\right)$ which is typically defined through the Fourier transform as

\begin{equation}
\int\limits_{-\infty}^{\infty} e^{- i k x} g_{\alpha,\beta}(x) \ud x \, = \, \exp(|k|^\alpha (1-i\beta \mathrm{sign}(k) \omega(k,\alpha))
\end{equation}
where $\omega(k,\alpha)= \tan \frac{\pi \alpha}{2}$ for $\alpha \neq 1$ and $\omega(k,1) = 2/\pi \ln |k|$.
 Parameter $\alpha \in [0,2]$ is the \textit{stability parameter}  which controls decay of the tails, and \textit{asymmetry parameter} $\beta \in [-1,1]$ that influences the asymmetry of the distribution. For $\beta=0$ we obtain symmetric distribution, for $\beta= \pm 1$ we have the distribution with maximal asymmetry. Sometimes, it is convenient to introduce the four-parameter class of stable distributions involving also location and scale parameters. This can be done straightforwardly $g_{\alpha,\beta;\mu,\sigma}(x) = g_{\alpha,\beta}\left(\frac{x-\mu}{\sigma}\right)$.

\noindent Stable distributions have heavy-tails, in other words they decay polynomially as $\frac{1}{|x|^\alpha}$ when $|x|\rightarrow\infty$, except in two cases \cite{Zolotarev86}:
\begin{itemize}
\item[$-$] $\alpha=2$, $g_{\alpha,\beta}$ reduces to the Gaussian distribution (regardless of $\beta$) and the L\'{e}vy process becomes ordinary Brownian diffusion;
\item[$-$] $\beta=\pm 1$, $g_{\alpha,\beta}$ possesses one heavy-tail (positive or negative, depending on the sign of $\beta$) and another tail with exponential decay for $\alpha > 1$. For $\alpha \leq 1$ the support of the distribution is even confined to the positive/negative half-axis.
\end{itemize}

\noindent In the case $\beta=-1$, the process $S_t$ defined by Eq. \eqref{Levy_SDE} has all moments finite. The finiteness of the moments is the result of the fact that the (two-sided) Laplace transform of $g_{\alpha,\beta}$ exists only for $\beta=-1$ and is equal to
\begin{equation}\label{laplace}
\int\limits_{-\infty}^\infty e^{-\lambda x} g_{\alpha,-1}(x) \ud x \, = \, \exp\left(- \frac{\lambda^\alpha}{\cos \frac{\pi \alpha}{2}} \right)
\end{equation}
The process with $\beta=-1$ has been used in connection to option pricing in \cite{Wu03} and is called Finite Moment Log Stable (FMLS) process. In this case, the solution for the SDE \eqref{Levy_SDE} is:
\begin{equation}
S_t \, = \, S_0 e^{(r + \mu)t} e^{\sigma L_{\alpha,-1}(t)}\, 
\end{equation}
The parameter $\mu$ comes from the Esscher transform $X \rightarrow \frac{e^{X}}{E[e^X]}$. The normalization $E[e^X]$ is connected to the existence of an equivalent risk-neutral measure, which has the martingale property \cite{Gerber93}. As a result, $\mu$ can be calculated as
\begin{equation}
\mu \, = \, \ln \int\limits_{-\infty} ^{+\infty} \, e^{x} \, g_{\alpha,\theta}(x) \, \id x\, . 
\end{equation}
The integral converges only for FMLS process, i.e. in the case only of maximal negative asymmetry case $\beta=-1$ \cite{Wu03} (compare with Eq. \eqref{laplace}) and reads
\begin{equation}
\mu \, = \, \frac{\left(\frac{\sigma}{\sqrt 2}\right)^\alpha}{\cos\frac{\pi\alpha}{2}}
\end{equation}
where we have introduced the $\sqrt{2}$ normalization, so that $\sigma$ becomes distribution variance of Gaussian distribution (Black-Scholes model -- see Appendix B for more details). 

\noindent
The price of an European call option can be generally expressed as
\begin{equation}\label{expect}
V(S_t,\tau) = e^{-r \tau } E[V(S_T,T)|S_t,t]
\end{equation} 
where $\tau = T-t$, $T$ is the maturity time. The boundary condition is for European call option defined as $V(S_T,T) = \max \{S_T-K,0\}:=[S_T-K]^+$, where $K$ is the strike price of the call option. For the FMLS model, the call price \eqref{expect} can be expressed as the convolution of a modified payoff and a Green function \cite{Kleinert16}:

\begin{equation}\label{Levy_Green}
V_{\alpha} (S,K,r,\mu,\tau) \, = \, \frac{e^{-r\tau}}{(-\mu \tau)^{\frac{1}{\alpha}}} \int\limits_{-\infty}^{+\infty}  [Se^{(r+\mu)\tau +y}-K]^+ g_{\alpha} \left(\frac{y}{(-\mu \tau)^{\frac{1}{\alpha}}}\right) \, \id y\, 
\end{equation}

\noindent 
where the Green function $g_{\alpha}$ can be expressed under the form of a Mellin-Barnes integral, that is, an integral over a vertical line in $\mathbb{C}$ (see Eq. \eqref{Green_FLMS_R+} in Appendix, where the reader will find more details and references about Mellin-Barnes integrals and their link to fractional diffusions): for any $X>0$,
\begin{equation}\label{Mellin_Green_LS}
g_{\alpha}(X) \, = \, \frac{1}{\alpha} \, \int\limits_{c_1-i\infty }^{c_1+ i \infty} \,
\frac{\Gamma(1-t_1)}{\Gamma(1-\frac{t_1}{\alpha})}
\, X^{t_1-1} \, \frac{\id t_1}{2i\pi}  \hspace*{1cm}  0 < c_1 < \min\{\alpha,1\}
\end{equation}

\noindent Let us mention that the asymmetry of the L\'{e}vy distribution $\beta$ can be equivalently described by the parameter $\theta$ that is confined to the so-called  \textit{Feller-Takayasu diamond} \cite{Gorenflo12}
\begin{equation}
\theta \leq \min\{ \alpha, 2-\alpha\}
\end{equation}
Parameter $\theta$ is uniquely determined by $\alpha$ and $\beta$. For example, for $\alpha > 1$, the value $\theta = -\alpha + 2$ corresponds to the case
when $\beta = 1$  and conversely, $\theta = \alpha - 2$ corresponds to $\beta = -1$. In all cases $\theta=0$ corresponds to $\beta=0$.  For more details about stable distributions, see e.g., Refs. \cite{Zolotarev86,Sato99,Kleinert09}.

\section{Mellin-Barnes representation for the stable call price}

\noindent We now derive an expression for the price \eqref{Levy_Green} under the form on an integral of some complex differential 2-form. Let then $S$, $K$, $\sigma$, $\tau$ $>0$; we will also assume that $\alpha>1$ (L\'evy-Pareto case) and that $\theta = \alpha - 2$ (Carr-Wu maximal negative asymmetry hypothesis).Let us define $[\log]:=\log\frac{S}{K} \, + \, r\tau$. We can rewrite the modified payoff as 
\begin{equation}
[Se^{(r+\mu)\tau +y}-K]^+ \, = \, K[e^{ [\log] + \mu\tau +y}-1]^+
\end{equation}
so that the option price can be conveniently rewritten as
\begin{equation}\label{option}
V_{\alpha} (S,K,r,\mu,\tau) \, = \, K\frac{e^{-r\tau}}{(-\mu \tau)^{\frac{1}{\alpha}}} \int\limits_{-\infty}^{+\infty}  [e^{[\log]+\mu\tau +y}-1]^+ g_{\alpha} \left(\frac{y}{(-\mu \tau)^{\frac{1}{\alpha}}}\right) \, \id y
\end{equation}
Let us now show how to rewrite the option price \eqref{option} 
with help of Mellin-Barnes integral representations.
\begin{lemma}
Assume that $[\log]+\mu\tau<0$; then for any $c_1<1$ the following holds:
\begin{equation}
V_{\alpha} (S,K,r,\mu,\tau) \, = \, \, K\frac{e^{-r\tau}}{\alpha} \, \int\limits_{c_1-i\infty}^{c_1+i\infty} \, \frac{\Gamma(1-t_1)}{\Gamma(1-\frac{t_1}{\alpha})}   \,  \int\limits_{-[\log]-\mu\tau}^{\infty} \, (e^{[\log]+\mu\tau +y}-1) \, y^{t_1-1} \id y \, (-\mu\tau)^{-\frac{t_1}{\alpha}} \, \frac{\id t_1}{2i\pi}
\end{equation}
\end{lemma}
\begin{proof}
We have just inserted \eqref{Mellin_Green_LS} into \eqref{option}. Note that, as we are in the L\'evy-Pareto case, $\min\{\alpha,1\}=1$.
\end{proof}
Now, introduce
\begin{equation}\label{Valpha_definition}
\left\{
\begin{aligned}
& V_{\alpha}^{(1)} (S,K,r,\mu,\tau) \, := \, K\frac{e^{-r\tau}}{\alpha} \, \int\limits_{c_1-i\infty}^{c_1+i\infty} \, \frac{\Gamma(1-t_1)}{\Gamma(1-\frac{t_1}{\alpha})} \int\limits_{-[\log]-\mu\tau}^{\infty} \, e^{[\log]+\mu\tau +y} \, y^{t_1-1} \id y \, (-\mu\tau)^{-\frac{t_1}{\alpha}} \, \frac{\id t_1}{2i\pi} \\
& V_{\alpha}^{(2)} (S,K,r,\mu,\tau) \, := \, K\frac{e^{-r\tau}}{\alpha} \, \int\limits_{c_1-i\infty}^{c_1+i\infty} \, \frac{\Gamma(1-t_1)}{\Gamma(1-\frac{t_1}{\alpha})} \int\limits_{-[\log]-\mu\tau}^{\infty} \,  y^{t_1-1} \id y \, (-\mu\tau)^{-\frac{t_1}{\alpha}} \, \frac{\id t_1}{2i\pi}
\end{aligned}
\right.
\end{equation}
so that we can write:
\begin{equation}\label{v12}
V_{\alpha} (S,K,r,\mu,\tau) \, = \, V_{\alpha}^{(1)} (S,K,r,\mu,\tau) \, - \, V_{\alpha}^{(2)} (S,K,r,\mu,\tau)
\end{equation}

\begin{proposition}
The following Mellin-Barnes representations in $\mathbb{C}^2$ and $\mathbb{C}$ hold:
\newline
\noindent (i) \underline{$V^{(1)}$-integral:}
\begin{equation}\label{Valpha1}
\begin{aligned}
& V_{\alpha}^{(1)} (S,K,r,\mu,\tau) \, = \, \\
& K\frac{e^{-r\tau}}{\alpha} \, \int\limits_{c_1-i\infty}^{c_1+i\infty}\int\limits_{c_2-i\infty}^{c_2+i\infty} \, (-1)^{-t_2} \, \frac{\Gamma(t_2)\Gamma(1-t_2)\Gamma(-t_1+t_2)}{\Gamma(1-\frac{t_1}{\alpha})} \left(-[\log]-\mu\tau \right)^{t_1-t_2} (-\mu\tau)^{-\frac{t_1}{\alpha}} \, \frac{\id t_1}{2i\pi} \wedge \frac{\id t_2}{2i\pi}
\end{aligned}
\end{equation}
and this double integral converges into the polyhedra $\left\{ Re(-t_1+t_2) > 0 , \, 0 < Re(t_2) < 1 \right\}$
\newline
\noindent (ii) \underline{$V^{(2)}$-integral:}
\begin{equation}\label{Valpha2}
V_{\alpha}^{(2)} (S,K,r,\mu,\tau) \, = \, K\frac{e^{-r\tau}}{\alpha} \, \int\limits_{c_1-i\infty}^{c_1+i\infty}  \, \frac{\Gamma(-t_1)}{\Gamma(1-\frac{t_1}{\alpha})} \left(-[\log]-\mu\tau \right)^{t_1} (-\mu\tau)^{-\frac{t_1}{\alpha}} \, \frac{\id t_1}{2i\pi}
\end{equation}
and this integral converges into the half-plane $\{Re(t_1)<0\}$
\end{proposition}
\begin{proof}
We start by (ii): integrating over the Green parameter $y$ in the $V^{(2)}$-integral in the definition \eqref{Valpha_definition} results in:
\begin{equation}
\int\limits_{-[\log]-\mu\tau}^{\infty} \, e^{[\log]+\mu\tau +y} \, y^{t_1-1} \id y \, = \, -\frac{(-[\log]-\mu\tau)^{t_1}}{t_1} \, \hspace*{2cm} [Re(t_1)<0]
\end{equation}
Then we use the Gamma functional relation:
\begin{equation}
\frac{\Gamma(1-t_1)}{t_1}=-\Gamma(-t_1)
\end{equation}
and this yields the desired result \eqref{Valpha2}.
\newline
\noindent To prove (ii), we start by introducing a supplementary Mellin-Barnes representation for the exponential term (see eq. \eqref{Cahen} in Appendix):
\begin{equation}
e^{[\log]+\mu\tau +y} \, = \, \int\limits_{c_2-i\infty}^{c_2+i\infty} \, (-1)^{-t_2} \, \Gamma(t_2) \, ([\log]+\mu\tau +y)^{-t_2} \, \frac{\id t_2}{2i\pi}
\end{equation}
where the integral converges for $Re(t_2)>0$, or, equivalently said, $c_2<0$. Plugging this into the $V^{(1)}$-integral in the definition \eqref{Valpha_definition} gives birth to a Beta integral \cite{Abramowitz72}:
\begin{equation}
\int\limits_{-[\log]-\mu\tau}^{\infty} ([\log] + \mu\tau + y)^{-t_2} \ y^{t_1-1} \, \id y  \, = \, \frac{\Gamma(1-t_2)\Gamma(-t_1+t_2)}{\Gamma(1-t_1)} \, (-[\log]-\mu\tau)^{t_1-t_2}  \hspace*{0.5cm} [ Re(t_2) < 1 \, , \, Re(-t_1+t_2)>0]
\end{equation}
Simplifying by $\Gamma(1-t_1)$ in the incoming double integral yields the representation \eqref{Valpha2}
\end{proof}

\section{Residue summation and closed pricing formula}

\subsection{$V^{(2)}$ integral}
\begin{proposition}
The $V_\alpha^{(2)}(S,K,r,\mu,\tau)$ integral can be expressed as the sum of the following absolutely convergent series:
\begin{equation}\label{Valpha2series}
V_\alpha^{(2)}(S,K,r,\mu,\tau) \, = \, \frac{Ke^{-r\tau}}{\alpha} \, \sum\limits_{n=0}^{\infty} \, \frac{(-1)^n}{n!\Gamma(1-\frac{n}{\alpha})} (-[\log]-\mu\tau)^{n}(-\mu\tau)^{-\frac{n}{\alpha}}
\end{equation}
\end{proposition}

\begin{proof}
The characteristic quantity associated to the Mellin-Barnes representation \eqref{Valpha2series}, is, by definition (see \cite{Tsikh94,Tsikh97} or Appendix)
is
\begin{equation}
\Delta \, = \, -1 \, + \, \frac{1}{\alpha} \, < \, 0
\end{equation}
\begin{figure}[t]
\centering
\includegraphics[scale=0.5]{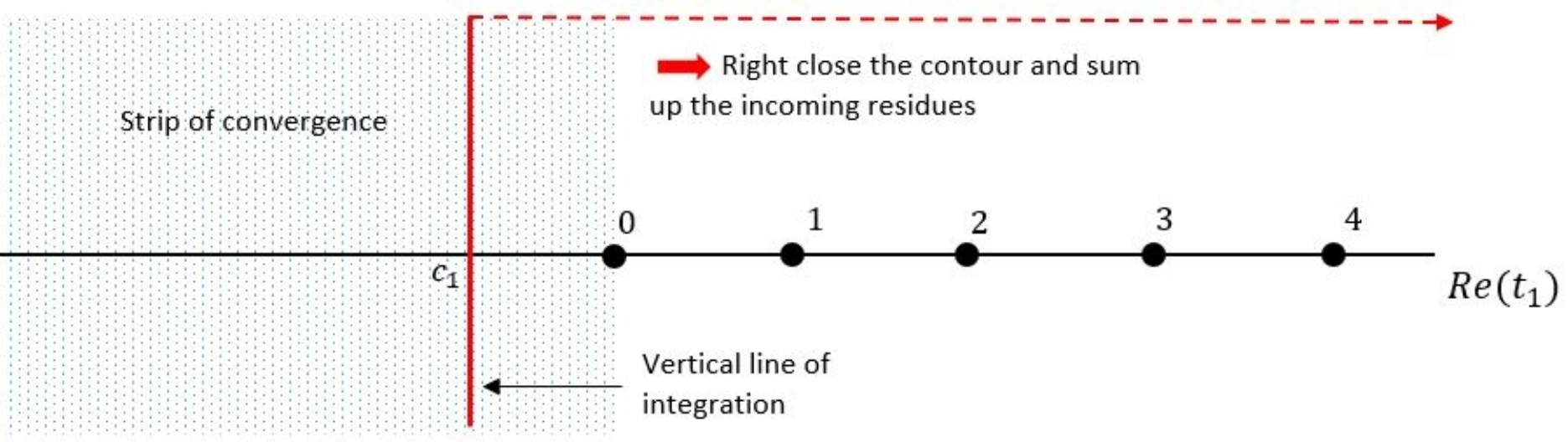}
\caption{Evaluating the $V^{(2)}-integral$ by summing the residues of its analytic continuation in the right-half plane}
\label{fig1}
\end{figure}
\noindent
as soon as $\alpha>1$. Therefore, the integral over the vertical line $(c_1-i\infty, c_1+ i \infty)$ can be right closed, thus corresponding to minus the sum of the residues associated to poles in the right half-plane (see fig. \ref{fig1}). Because of the singular behavior of the Gamma function around its negative arguments \eqref{sing_Gamma} these residues are induced by the $\Gamma(-t_1)$-term and equal
\begin{equation}
\res(t_1 \, = \, +n) \, = \, -\frac{Ke^{-r\tau}}{\alpha} \, \frac{(-1)^n}{n!} \, \frac{1}{\Gamma(1-\frac{n}{\alpha})} \, (-[\log]-\mu\tau)^n \, (-\mu\tau)^{-\frac{n}{\alpha}}
\end{equation}
\end{proof}

\subsection{$V^{(1)}$ integral}
\begin{proposition}
The $V_\alpha^{(1)}(S,K,r,\mu,\tau)$ integral can be expressed as the sum of the following absolutely convergent double series:
\begin{equation}\label{Valpha1series}
V_\alpha^{(1)}(S,K,r,\mu,\tau) \, = \, \frac{Ke^{-r\tau}}{\alpha} \, \sum\limits_{n,m=0}^{\infty} \, \frac{(-1)^n}{n!\Gamma(1-\frac{n-m}{\alpha})} (-[\log]-\mu\tau)^{n}(-\mu\tau)^{\frac{n-n}{\alpha}}
\end{equation}
\end{proposition}

\begin{proof}
Let us introduce the notations
\begin{equation}
\underline{c} \ :=
\begin{bmatrix}
c_1 \\ c_2
\end{bmatrix}
\hspace*{1cm}
\underline{t} \ :=
\begin{bmatrix}
t_1 \\ t_2
\end{bmatrix}
\hspace*{1cm}
d\underline{t} \, := \, dt_1 \wedge dt_2
\end{equation}
and the complex differential 2-form
\begin{equation}
\omega \, := \, (-1)^{-t_2} \, \frac{\Gamma(t_2)\Gamma(1-t_2)\Gamma(-t_1+t_2)}{\Gamma(1-\frac{t_1}{\alpha})} \left(-[\log]-\mu\tau \right)^{t_1-t_2} (-\mu\tau)^{-\frac{t_1}{\alpha}} \, \frac{\id t_1}{2i\pi} \wedge \frac{\id t_2}{2i\pi}
\end{equation}
so that \eqref{Valpha1} can be written under the standard form of a Mellin-Barnes integral in $\mathbb{C}^2$:
\begin{equation}
V_\alpha^{(1)}(S,K,r,\mu,\tau) \, = \, \frac{Ke^{-r\tau}}{\alpha} \, \int\limits_{\underline{c}+i\mathbb{R}^2} \, \omega
\end{equation}
The characteristic vector associated to $\omega$ is (see \cite{Tsikh94,Tsikh97} and Appendix of this paper) :
\begin{equation}
\Delta \, = \,
\begin{bmatrix}
-1 + \frac{1}{\alpha} \\ 1
\end{bmatrix}
\end{equation}
and therefore the half-plane of convergence one must considers:
\begin{equation}
\Pi_\Delta \, := \, \, \, \, \, \, \{ \underline{t} \in \mathbb{C}^2, \, Re( \Delta \, . \, \underline{t}) \,  \, < \,  \, Re( \Delta \, . \, \underline{c} \, ) \}
\end{equation}
is the one located under the line (see fig. \ref{fig2}):
\begin{equation}
Re(t_2) \, = \, (1-\frac{1}{\alpha}) Re(t_1) \, + \, (-1+\frac{1}{\alpha})c_1 + c_2
\end{equation}
\begin{figure}[t]
\centering
\includegraphics[scale=0.5]{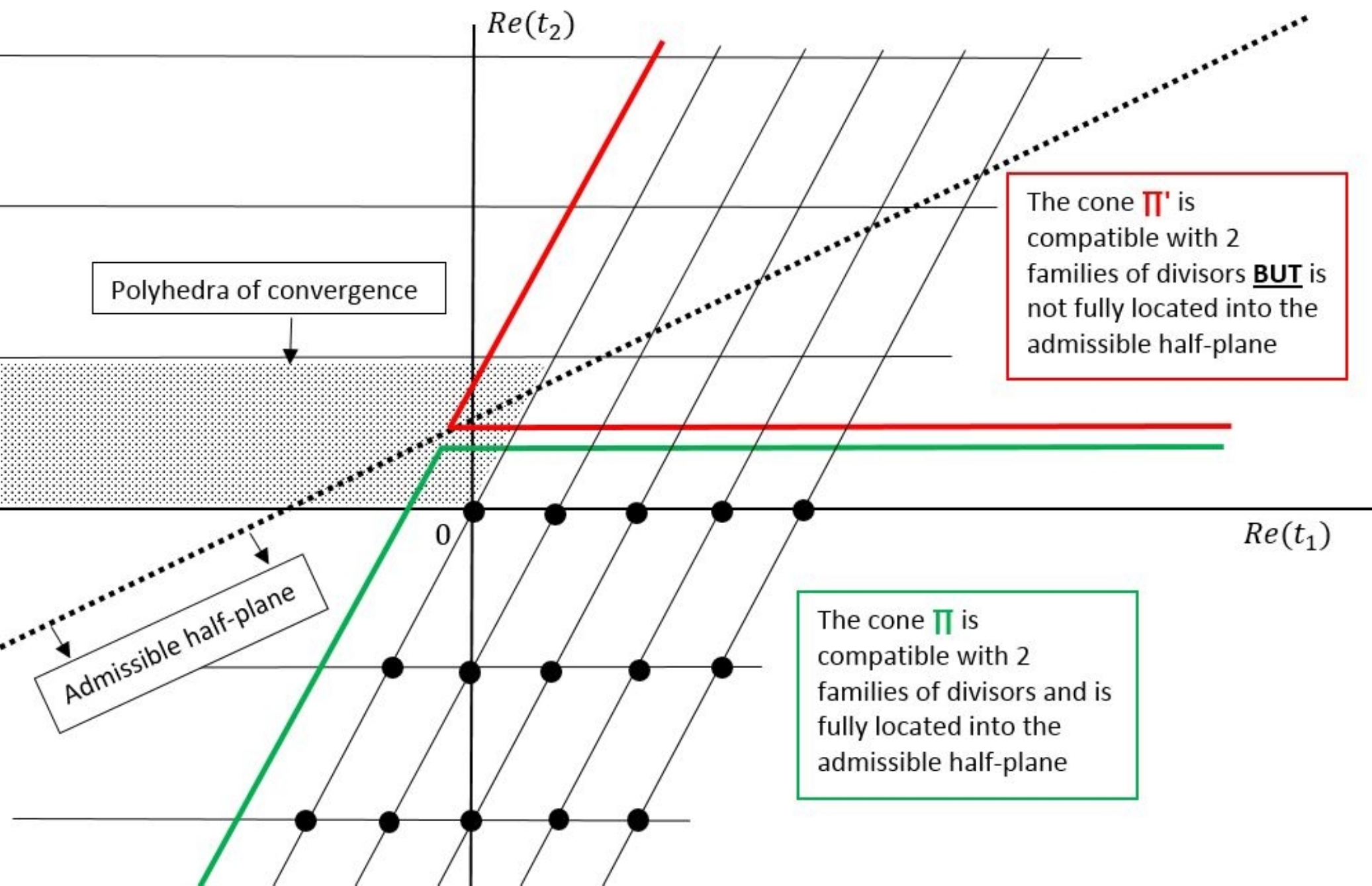}
\caption{The admissible region $\Pi_{\Delta}$, for the complex 2-form $\omega$, is the one located under the dotted oblique line. There is only one compatible cone in this region: the green cone, which is compatible with the two family of divisors $D_1$ (oblique lines) and $D_2$ (horizontal lines). The sum of the residues at every $\mathbb{C}^2$-singularity (points) into this cone is therefore equal to the integral of $\omega$.}
\label{fig2}
\end{figure}
Because $0 < 1-\frac{1}{\alpha} < 1$, the cone $\Pi$ defined by
\begin{equation}
\Pi \, := \,\,\,\, \{\underline{t} \in \mathbb{C}^2 \, , \, Re(t_2) < 0 \, , \, Re(-t_1+t_2) < 0 \}
\end{equation}
is included in $\Pi_\Delta$; moreover it is compatible with the two families of divisors
\begin{equation}
\left\{
\begin{aligned}
& D_1 \, = \, \left\{\underline{t}\in\mathbb{C}^2,  - t_1 + t_2 = -n_1 \,\, , \,\,\, n_1 \in\mathbb{N} \right\}
\\
& D_2 \, = \, \left\{\underline{t}\in\mathbb{C}^2, t_2 = -n_2 \,\, , \,\,\, n_2 \in\mathbb{N} \right\}
\end{aligned}
\right.
\end{equation}
induced by the $\Gamma(-t_1+t_2)$ and $\Gamma(t_2)$ functions respectively. To compute the residues associated to every element of the singular set $D := D_1 \cap D_2$, we change the variables:
\begin{equation}
\left\{
\begin{aligned}
& u_1 \, := \,-t_1 + t_2 \\
& u_2 \, := \, t_2
\end{aligned}
\right.
\longrightarrow
\left\{
\begin{aligned}
& t_1 \, = \, u_2-u_1 \\
& t_2 \, = \, u_2 
\end{aligned}
\right.
\end{equation}
so that in this new configuration $\omega$ reads
\begin{equation}
\omega \, = \, (-1)^{-u_2} \, \frac{\Gamma(u_2)\Gamma(1-u_2)\Gamma(u_1)}{\Gamma(1-\frac{u_2-u_1}{\alpha})} \left(-[\log]-\mu\tau \right)^{-u_1} (-\mu\tau)^{\frac{u_1-u_2}{\alpha}} \, \frac{\id u_1}{2i\pi} \wedge \frac{\id u_2}{2i\pi}
\end{equation}
With this new variables, the divisors $D_1$ and $D_2$ are induced by the $\Gamma(u_1)$ and $\Gamma(u_2)$ functions, and intersect at every point of the type $(u_1,u_2)=(-n,-m)$, $n,m\in\mathbb{N}$. From the singular behavior of the Gamma function \eqref{sing_Gamma} around a singularity, we can write:
\begin{equation}
\omega \, \underset{(u_1,u_2)\rightarrow (-n,-m)} {\sim} \,  \frac{(-1)^{n+m}}{n!m!} (-1)^{-u_2} \, \frac{\Gamma(1-u_2)}{\Gamma(1-\frac{u_2-u_1}{\alpha})} \left(-[\log]-\mu\tau \right)^{-u_1} (-\mu\tau)^{\frac{u_1-u_2}{\alpha}} \, \frac{\id u_1}{2i\pi (u_1+n)} \wedge \frac{\id u_2}{2i\pi (u_2 + m)}
\end{equation}
Taking the residues and simplifying:
\begin{equation}
\res_{(n,m)} \, \omega \, = \,  \frac{(-1)^n}{n!\Gamma(1-\frac{n-m}{\alpha})} \left(-[\log]-\mu\tau \right)^{n} (-\mu\tau)^{\frac{m-n}{\alpha}}
\end{equation}
Summing in the whole cone $\Pi$  yields the series \eqref{Valpha1series}
\end{proof}

\subsection{Closed pricing formula}

\begin{theorem}
Let $S,K,r,\tau,\sigma>0$, let $[\log]:=\log\frac{S}{K}+r\tau$ and $\mu \, := \, \frac{(\sigma/\sqrt{2})^\alpha}{\cos\frac{\pi\alpha}{2}}$. For any $1<\alpha \leq 2$ and under maximal asymmetry hypothesis, the European call price is equal to:
\begin{equation}\label{Formula}
V_{\alpha}(S,K,r,\mu,\tau) \, = \, \frac{Ke^{-r\tau}}{\alpha} \, \sum\limits_{\substack{n = 0 \\ m = 1}}^{\infty} \, \frac{(-1)^n}{n!\Gamma(1-\frac{n-m}{\alpha})} (-[\log]-\mu\tau)^{n}(-\mu\tau)^{\frac{m-n}{\alpha}}
\end{equation}
\end{theorem}
\begin{proof}
Let us just recall Eq. \eqref{v12} and use the series obtained in \eqref{Valpha2series} and \eqref{Valpha1series}. Let us remark that the series \eqref{Valpha2series} is indeed a particular case of \eqref{Valpha1series} for $m=0$, and the proof of the pricing formula is complete.
\end{proof}

\noindent
In the case of Black-Scholes, i.e. $\alpha = 2$, the pricing formula \eqref{Formula} can be rewritten as:
\begin{equation}
V_{BS}(S,K,r,\mu,\tau) \, = \, \frac{Ke^{-r\tau}}{2} \, \sum\limits_{\substack{n = 0 \\ m = 1}}^{\infty} \, \frac{(-1)^n}{n!\Gamma(1-\frac{n-m}{2})} (-[\log]-\mu\tau)^{n}(-\mu\tau)^{\frac{m-n}{2}}
\end{equation}
where in this case $\mu = -\frac{\sigma^2}{2}$. We recover the formula obtained in \cite{Aguilar17} for the Black-Scholes call. In particular, if we suppose that the asset is "at-the-money forward", that is:
\begin{equation}
S \, = \, K e ^{-r\tau}
\end{equation}
then by definition $[\log] = 0$ and we are left with:
\begin{equation}\label{Brenner}
V^{ATMF}_{BS}(S,K,r,\sigma,\tau) \, = \, \frac{S}{2} \, \sum\limits_{\substack{n = 0 \\ m = 1}}^{\infty} \, \frac{(-1)^n}{n!\Gamma(1-\frac{n-m}{2})} \, \left( \sigma \frac{\sqrt{\tau}}{\sqrt{2}}  \right) ^{n+m}
\end{equation}
Series \eqref{Brenner} is a power series of the market volatility which starts for $n=0, m=1$:
\begin{equation}
V^{ATMF}_{BS}(S,K,r,\sigma,\tau) \, = \, \, \frac{S}{2} \left[ \, \frac{1}{\Gamma(\frac{3}{2})} \sigma \frac{\sqrt{\tau}}{\sqrt{2}} \, + \, \dots  \, \right]
\end{equation}
Recalling $\Gamma(\frac{3}{2}) = \frac{\sqrt{\pi}}{2}$, we have:
\begin{equation}
V^{ATMF}_{BS}(S,K,r,\sigma,\tau) \, = \, \frac{1}{\sqrt{2\pi}} S \sigma \sqrt{\tau} \dots \, + \, \simeq \, 0.4 S \sigma \sqrt{\tau}
\end{equation}
which is the Brenner-Subrahmanyam approximation \cite{BS94} for the call price. Series \eqref{Brenner} is also derived in appendix directly from Taylor expanding the Black-Scholes formula, which turns out to be easy in the ATM-froward case. The interested reader can compare the two equivalent representations \eqref{Brenner} and \eqref{Brenner_Series_2}.

\section{Numerical applications}



\noindent Let us demonstrate the results obtained in the previous section in the real estimation of option pricing. First, we demonstrate that the double-series \eqref{Formula} converges quickly after very few terms. Tab. \ref{fig:series} provides an example of the series convergence of an option with parameters $S=3800, K=4000, r=1\%, \sigma= 20\%, \tau = 1, \alpha = 1.7$. The convergence is a little bit slower when $\tau$ grows, but remains very efficient, see for instance Table \ref{fig:series5Y}. In Fig \ref{fig7} we compute the partial terms of the series \eqref{Formula} as a series of $m$, that is, the sum of vertical columns in tables \ref{fig:series} and \ref{fig:series5Y}, for a time to maturity $\tau$ being 1 or 5 years. We observe that, in both cases, the convergence is fast.
\begin{table}[h!]
\centering
\begin{scriptsize}
\begin{tabular}{|c||ccccccc|}
  \hline
 {\bfseries $\tau$=1Y}  & 1 & 2 & 3 & 4 & 5 & 6 & 7   \\
  \hline
  \hline
  0 & 395.167 & 49.052 & 4.962 & 0.431 & 0.033 & 0.002 & 0.000     \\
  1 & -190.223 & -32.268 & -4.005 & -0.405 & -0.035 & -0.003 & -0.000   \\
  2 & 23.829 & 7.767 & 1.317 & 0.164 & 0.017 & 0.001 & 0.000    \\
  3 & 1.430 & -0.649 & -0.211 & -0.036 & -0.004 & -0.000 & -0.000   \\
  4 & -0.246 & -0.029 & 0.013 & 0.001 & 0.000 & 0.000 & 0.000   \\
  5 & -0.046 & 0.004 & 0.000 & -0.000 & -0.000 & -0.000 & -0.000  \\
  6 & 0.001 & 0.000 & -0.000 & -0.000 & 0.000 & 0.000 & 0.000   \\
  7 & 0.001 & -0.000 & -0.000 & 0.000 & 0.000 & -0.000 & -0.000   \\
  8 & 0.000 & -0.000 & 0.000 & 0.000 & -0.000 & -0.000 & 0.000   \\
  \hline
  Call & 229.914 & 253.790 & 255.866 & 256.024 & 256.035 & 256.035 & 256.035    \\
  \hline
\end{tabular}
\end{scriptsize}
\caption{Table containing the numerical values for the $(n,m)$-term in the series (\ref{Formula}) for the option price ($S=3800, \, K=4000, \, r=1\%, \sigma=20\%, \, \tau=1Y, \, \alpha=1.7$). The call price converges to a precision of $10^{-3}$ after summing only very few terms of the series.}
\label{fig:series}
\end{table}

\begin{table}[h!]
\centering
\begin{scriptsize}
\begin{tabular}{|c||cccccccccc|}
  \hline
 {\bfseries $\tau$=5Y}  & 1 & 2 & 3 & 4 & 5 & 6 & 7 &8 & 9 & 10  \\
  \hline
  \hline
  0 & 978.516 & 313.038 & 81.607 & 18.274 & 3.626 & 0.651 & 0.108  & 0.016 & 0.002 & 0.000   \\
  1 & -454.606 & -198.750 & -63.582 & -16.576 & -3.712 & -0.736 & -0.132 & -0.022 & -0.003 & -0.000  \\
  2 & 54.963 & 46.168 & 20.184 & 6.457 & 1.683 & 0.377 & 0.075  & 0.013 & 0.002 & 0.000  \\
  3 & 3.183 & -3.721 & -3.1258 & -1.367 & -0.438 & -0.114 & -0.026  & -0.005 & -0.001 & -0.000  \\
  4 & -0.529 & -0.162 & 0.189 & 0.158 & 0.069 & 0.022 & 0.006 & 0.001 & 0.000 & 0.000 \\
  5 & -0.096 & 0.021 & 0.007 & -0.008 & -0.006 & -0.003 & -0.001 & 0.000 & 0.000 & 0.000  \\
  6 & 0.003 & 0.003 & -0.001 & -0.000 & 0.000 & 0.000 & 0.000 & 0.000 & 0.000 & 0.000  \\
  7 & 0.002 & -0.000 & -0.000 & 0.000 & 0.000 & -0.000 & -0.000 & -0.000 & -0.000 & -0.000  \\
  8 & -0.000 & -0.000 & 0.000 & -0.000 & -0.000 & 0.000 & 0.000 & -0.000 & -0.000 & -0.000  \\
  \hline
  Call & 581.436 & 738.034 & 773.313 & 780.252 & 781.475 & 781.672 & 781.702 & 781.706 & 781.706 &  781.706 \\
  \hline
\end{tabular}
\end{scriptsize}
\caption{Same set of parameters as in Table \ref{fig:series}, but for a longer maturity $\tau$= 5 years. The convergence is slightly slower but remains very efficient.}
\label{fig:series5Y}
\end{table}

 \begin{figure}[h]
 \centering
 \includegraphics[scale=0.3]{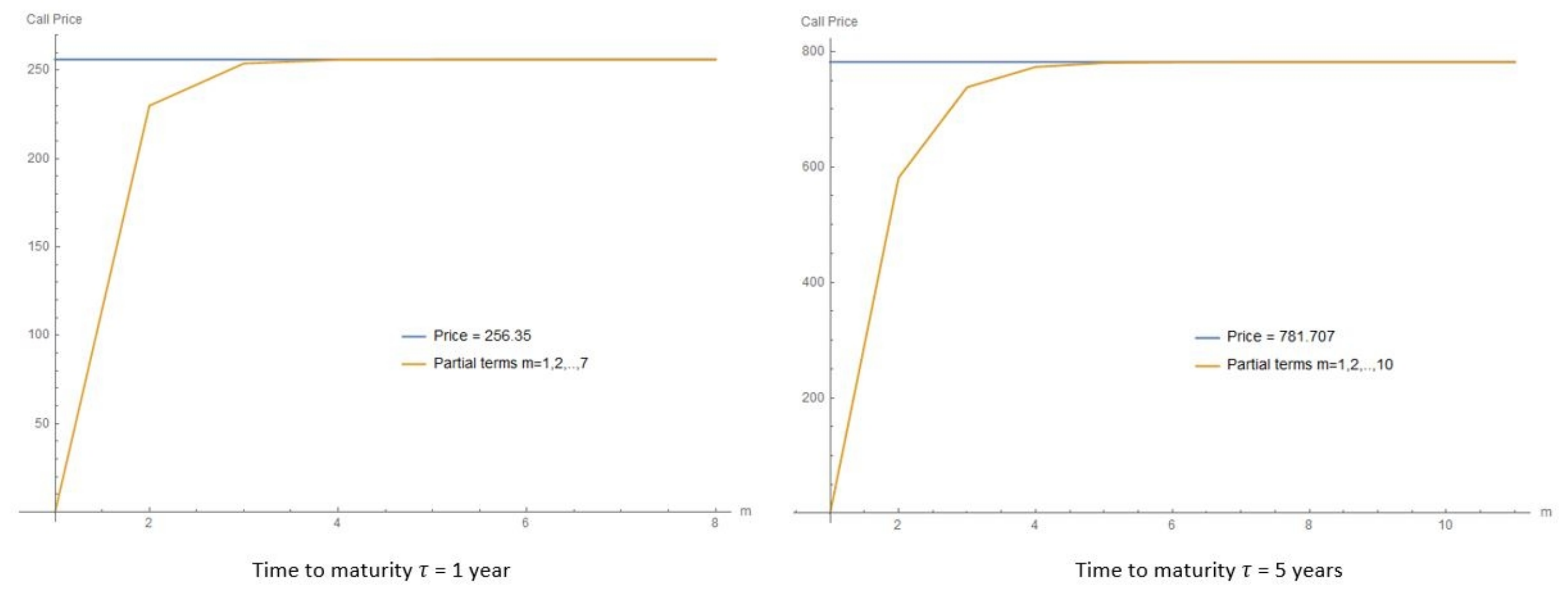}
\caption{Convergence of partial terms of the series \eqref{Formula} to the call price, for short and long times to maturity (1Y and 5Y).}
 \label{fig7}
 \end{figure}


Additionally, we compare the prices obtained for various L\'evy parameter $\alpha$ and market configurations (in and out of the money), via different techniques. Let us assume an option with market parameters $K=4000, r=1\%, \sigma= 20\%, \tau = 1$. We compare in and out of the money option prices via three different techniques:
\begin{itemize}
\item \textbf{Series formula \eqref{Formula}} obtained in this paper. Only the first terms of the double series are taken into account ($n=0\dots 4$, $m=1 \dots 5$), demonstrating its quickness of convergence;
\item \textbf{Gil-Pelaez numerical estimation}. The results obtained by a numerical evaluation of the Gil-Pelaez type: the price of an European call option with spot $S$, strike $K$ and maturity $T$ is:
\begin{align*}
C(S,K,T)&=e^{-rT}\mathbb{E}^\mathbb{Q}[(S-K)^+] \\
&= S_T \Pi_1-e^{-rT}K\Pi_2
\end{align*}
Using Gil-Pelaez (see for instance \cite{Rouah13}) and denoting $k=\log K$ one finds
\begin{equation}
\Pi_1=\frac{1}{2}+\frac{1}{\pi}\int_0^{\infty}\mathcal{R}\left[\frac{e^{iuk}\phi_T(u-i)}{iu}\right] \id u
\end{equation}
\begin{equation}
\Pi_2=\frac{1}{2}+\frac{1}{\pi}\int_0^{\infty}\mathcal{R}\left[\frac{e^{iuk}\phi_T(u)}{iu}\right] \id u
\end{equation}
Finally, rearranging the equation of the call price, we get
\begin{multline}\label{GilPelaez}
C(S,K,T) \,= \, \\
\frac{1}{2}(S - K e^{-rT}) + 
\frac{1}{\pi} \left[ S \, \int_0^{\infty} \mathcal{R}\left[\frac{e^{iuk}\phi_T(u-i)}{iu}\right]\id u -Ke^{-rT}\int_0^{\infty}\mathcal{R}\left[\frac{e^{iuk}\phi_T(u)}{iu}\right] \id u \right]
\end{multline}
The characteristic function in the L\'evy-stable model with maximal assymetry is \cite{Wu03}:
\begin{equation}
\phi_T(u)=\mathbb{E}[e^{iuX_T}]=\exp\left[\mu T(iu-(iu)^\alpha) \right]
\end{equation}
where $\mu=(\frac{\sigma}{\sqrt{2}})^\alpha\sec\left(\frac{\pi\alpha}{2}\right)$; the integrals in \eqref{GilPelaez} can be carried out easily with a code in R or C++.  In the following array, we show that our analytic results are completely confirmed by this evaluation for any $\alpha$.
\item \textbf{Discretization of $g_{\alpha}$} calculated by Mellin-Barnes integral used in \eqref{Levy_Green}. One calculates the probability distribution by the Mellin-Barnes integral for a grid of functions, and replace the integral for option pricing \eqref{Levy_Green} by a discrete sum over this grid. Of course the price is as more precise as the grid is dense. In the following array, we observe that the result of this technique fits very well with the analytic results for $\alpha > 1.6$; this techniques becomes less precise when $\alpha<1.6$: this is because the probability distribution function decays less fast, and thus the grid should be extended to become more accurate.
\end{itemize}

\vspace*{0.5cm}
\noindent
\begin{center}
\begin{tabular}{|c|c|c|c|c|c|c|}
\hline
\textbf{Out of the money (S=3800)}& $\alpha = 1.5$ & $\alpha = 1.6$ & $\alpha = 1.7$ & $\alpha = 1.8$ & $\alpha = 1.9$ & $\alpha = 2$ \\
\hline
\hline
\textbf{Analytic result \eqref{Formula}} & 284.52 & 268.52 &  256.04  & 246.60  & 239.83 & 235.52 \\
\hline
\textbf{Gil-Pelaez} & 284.52& 268.52 & 256.04  & 246.59  & 239.83  & 235.51  \\
\hline
\textbf{Discretization} & 292.74 & 272.17 & 257.63   &  247.45  & 240.20  & 235.74   \\
\hline
\hline
\hline
\textbf{In the money (S=4200)}& $\alpha = 1.5$ & $\alpha = 1.6$ & $\alpha = 1.7$ & $\alpha = 1.8$ & $\alpha = 1.9$ & $\alpha = 2$ \\
\hline
\hline
\textbf{Analytic result \eqref{Formula}} & 547.67 & 523.25 &  502.53  & 485.07  & 470.56  & 458.79 \\
\hline
\textbf{Gil-Pelaez} & 547.67  & 523.25 &  502.53 & 485.07 & 470.56  & 458.79  \\
\hline
\textbf{Discretization} & 557.47 & 527.39 & 504.59 & 486.05 & 470.98 & 459.10  \\
\hline
\end{tabular}
\end{center}

\section{Conclusion and perspectives}
In this paper, we have shown that the formula for an European call option driven by the Finite moment log-stable process can be obtained in the form of rapidly convergent double series (\ref{Formula}).

\noindent
The formula provides a new and efficient analytic tool in the non-Gaussian derivative pricing -- it provides the market practitioner with a simple tool that can be easily used without any deeper knowledge of advanced mathematical techniques.
The series terms are straightforward to evaluate by simple functions - the only "transcendence" being carried by the particular values of the Gamma function.
It can also be easily used to imply the volatility from market observations.

\noindent The convergence of the series is fast enough to ensure excellent levels of precision only with summing the first few terms of the series. Moreover, as there is no need to use a numerical estimation techniques.  

\noindent Moreover, the residue technique can be successfully applied to a much wider class of models. First, one would like to relax the maximal asymmetry hypothesis $\theta = \alpha - 2$, which may not be precise enough for less liquid markets; in this case, the L\'evy propagator \eqref{Levy_Green} is known to diverge, but the Mellin-Barnes integral representation provide us with a natural regularization tool, and the residues summation will remain finite in this case. Second, one can extend the fractional differential equation \eqref{fract_diff} to a \textit{space-time fractional diffusion}. The resulting Green function is obtained as a superposition of stable distributions. In both cases, the residue summation technique will give a finite result. The computation will be detailed in future works.

\section*{Acknowledgements}
\noindent We wish to thank Raphael Douady and Zari Rachev for their interest in this work and their valuable suggestions. J. K. was supported by the Austrian Science Fund, grant No. I 3073-N32 and the Czech Science Foundation, grant No. 17-33812L.

\newpage
\appendix
\section{APPENDIX: Mellin transforms and residues}

We briefly present here some of the concepts used in the paper. The theory of the one-dimensional Mellin transform is explained in full detail in \cite{Flajolet95}. An introduction to multidimensional complex analysis can be found in the classic textbook \cite{Griffiths78}, and applications to the specific case of Mellin-Barnes integrals is developed in \cite{Tsikh94,Tsikh97}.

\subsection{One-dimensional Mellin transforms}

1. The Mellin transform of a locally continuous function $f$ defined on $\mathbb{R}^+$ is the function $f^*$ defined by
\begin{equation}\label{Mellin_def}
f^*(s) \, := \, \int\limits_0^\infty \, f(x) \, x^{s-1} \, \id x
\end{equation}
The region of convergence $\{ \alpha < Re (s) < \beta \}$ into which the integral \eqref{Mellin_def} converges is often called the fundamental strip of the transform, and sometimes denoted $ < \alpha , \beta  > $.
\newline
\noindent 2. The Mellin transform of the exponential function is, by definition, the Euler Gamma function:
\begin{equation}
\Gamma(s) \, = \, \int\limits_0^\infty \, e^{-x} \, x^{s-1} \, \id x
\end{equation}
with strip of convergence $\{ Re(s) > 0 \}$. Outside of this strip, it can be analytically continued, expect at every negative integer $s=-n$ where it admits the singular behavior
\begin{equation}\label{sing_Gamma}
\Gamma(s) \, \underset{s\rightarrow -n}{\sim} \, \frac{(-1)^n}{n!}\frac{1}{s+n} \hspace*{1cm} n\in\mathbb{N}
\end{equation}
\newline
\noindent 3. The inversion of the Mellin transform is performed via an integral along any vertical line in the strip of convergence:
\begin{equation}\label{inversion}
f(x) \, = \, \int\limits_{c-i\infty}^{c+i\infty} \, f^*(s) \, x^{-s} \, \frac{\id s}{2i\pi} \hspace*{1cm} c\in ( \alpha, \beta )
\end{equation}
and notably for the exponential function one gets the so-called \textit{Cahen-Mellin integral}:
\begin{equation}\label{Cahen}
e^{-x} \, = \, \int\limits_{c-i\infty}^{c+i\infty} \, \Gamma(s) \, x^{-s} \, \frac{\id s}{2i\pi} \hspace*{1cm} c>0
\end{equation}
\newline
\noindent 4. When $f^*(s)$ is a ratio of products of Gamma functions of linear arguments:
\begin{equation}
f^*(s) \, = \, \frac{\Gamma(a_1 s + b_1) \dots \Gamma(a_n s + b_n)}{\Gamma(c_1 s + d_1) \dots \Gamma(c_m s + d_m)}
\end{equation}
then one speaks of a \textit{Mellin-Barnes integral}, whose \textit{characteristic quantity} is defined to be
\begin{equation}
\Delta \, = \, \sum\limits_{k=1}^n \, a_k \, - \, \sum\limits_{j=1}^m \, c_j
\end{equation}
$\Delta$ governs the behavior of $f^*(s)$ when $|s|\rightarrow \infty$ and thus the possibility of computing \eqref{inversion} by summing the residues of the analytic continuation of $f^*(s)$ right or left of the convergence strip:
\begin{equation}
\left\{
\begin{aligned}
& \Delta < 0 \hspace*{1cm} f(x) \, = \, -\sum\limits_{Re(s_N) > \beta} \, \res_{S_N} f^*(s)x^{-s}  \\
& \Delta > 0 \hspace*{1cm} f(x) \, = \, \sum\limits_{Re(s_N) < \alpha} \, \res_{S_N} f^*(s)x^{-s}
\end{aligned}
\right.
\end{equation}
For instance, in the case of the Cahen-Mellin integral one has $\Delta = 1$ and therefore:
\begin{equation}
e^{-x} \, = \, \sum\limits_{Re(s_n)<0} \res_{s_n} \Gamma(s) \, x^{-s} \, = \, \sum\limits_{n=0}^{\infty} \, \frac{(-1)^n}{n!}x^n
\end{equation}
as expected from the usual Taylor series of the exponential function.

\subsection{Multidimensional Mellin transforms}

1. Let the $\underline{a}_k$, $\underline{c}_j$, be vectors in $\mathbb{C}^2$,and the $b_k$, $d_j$ be complex numbers. Let $\underline{t}:=\begin{bmatrix} t_1 \\ t_2 \end{bmatrix}$ and $\underline{c}:=\begin{bmatrix} c_1 \\ c_2 \end{bmatrix}$ in $\mathbb{C}^2$ and "." represent the euclidean scalar product. We speak of a Mellin-Barnes integral in $\mathbb{C}^2$ when one deals with an integral of the type
\begin{equation}
\int\limits_{\underline{c}+i\mathbb{R}^2} \, \omega
\end{equation}
where $\omega$ is a complex differential 2-form who reads
\begin{equation}
\omega \, = \, \frac{\Gamma(\underline{a}_1.\underline{t}_1 + b_1) \dots \Gamma(\underline{a}_n.\underline{t}_n + b_n)}{\Gamma(\underline{c}_1.\underline{t}_1 + d_1) \dots \Gamma(\underline{c}_m.\underline{t}_m + b_m)} \, x^{-t_1} \, y^{-t_2} \, \frac{\id t_1}{2i\pi} \wedge \frac{\id t_2}{2i\pi} \hspace*{1cm} \, x,y \in\mathbb{R}
\end{equation}
The singular sets induced by the singularities of the Gamma functions
\begin{equation}
D_k \, := \, \{ \underline{t}\in\mathbb{C}^2 \, , \, \underline{a}_k.\underline{t}_k + b_k = -n_k \, , \, n_k \in\mathbb{N}   \} \,\,\,\, \, k=0 \dots n
\end{equation}
are called the \textit{divisors} of $\omega$. The \textit{characteristic vector} of $\omega$ is defined to be
\begin{equation}
\Delta \, = \, \sum\limits_{k=1}^n \underline{a}_k \, - \, \sum\limits_{j=1}^m \underline{c}_j
\end{equation}
and the \textit{admissible half-plane}:
\begin{equation}
\Pi_\Delta \, := \, \{ \underline{t}\in\mathbb{C}^2 \, , \, Re( \Delta . \underline{t} ) \, < \, Re( \Delta . \underline{c} )\,  \}
\end{equation}
\newline
\noindent 2. Let the $\rho_k$ in $\mathbb{R}$, the $h_k:\mathbb{C}\rightarrow\mathbb{C}$ be linear aplications and $\Pi_k$ be a subset of $\mathbb{C}^2$ of the type
\begin{equation}\label{Pik}
\Pi_k \, := \, \{ \underline{t}\in\mathbb{C}^2, \, Re(h_k(t_k)) \, < \, \rho_k \}
\end{equation}
A \textit{cone} in $\mathbb{C}^2$ is a cartesian product
\begin{equation}
\Pi \, = \, \Pi_1 \times \Pi_2
\end{equation}
where $\Pi_1$ and $\Pi_2$ are of the type \eqref{Pik}. Its \textit{faces} $\varphi_k$ are
\begin{equation}
\varphi_k \, := \, \partial \Pi_k \hspace*{1cm} k=1,2
\end{equation}
and its \textit{distinguished boundary}, or \textit{vertex} is
\begin{equation}
\partial_0 \, \Pi \, := \, \varphi_1 \, \cap \, \varphi_2
\end{equation}
\newline
3. Let $1<n_0<n$. We group the divisors $D=\cup_{k=0}^n \, D_k$ of the complex differential form $\omega$ into two sub-families
\begin{equation}
D_1 \, := \, \cup_{k=1}^{n_0} \, D_k \,\,\, \,\,\, D_2 \, := \, \cup_{k=n_0+1}^{n} \, D_k  \hspace*{1cm}  D \, = \, D_1\cup D_2
\end{equation}
We say that a cone $\Pi\subset\mathbb{C}^2$ is \textit{compatible} with the divisors family $D$ if:
\begin{enumerate}
\item[-] \, Its distinguished boundary is $\underline{c}$;
\item[-] \, Every divisor $D_1$ and $D_2$ intersect at most one of his faces:
\begin{equation}
D_k \, \cap \, \varphi_k \, = \, \emptyset \hspace*{1cm} k=1,2
\end{equation}
\end{enumerate}

\noindent 4. Residue theorem for multidimensional Mellin-Barnes integral \cite{Tsikh94,Tsikh97}: If $\Delta \neq 0$ and if $\Pi\subset\Pi_\Delta$ is a compatible cone located into the admissible half-plane, then
\begin{equation}
\int\limits_{\underline{c}+i\mathbb{R}^2} \, \omega \, = \, \sum\limits_{\underline{t}\in\Pi\cap(D_1 \cap D_2)} \res_{\underline{t}} \, \omega
\end{equation}
and the series converges absolutely. The residues are to be understood as the "natural" generalization of the Cauchy residue, that is:
\begin{equation}
\res_0 \, \left[ f(t_1,t_2) \, \frac{\id t_1}{2i\pi t_1^{n_1}} \wedge \frac{\id t_2}{2i\pi t_1^{n_2}}  \right] \, = \, \frac{1}{(n_1-1)!(n_2-1)!}\frac{\partial ^{n_1+n_2-2}}{\partial t_1^{n_1-1} \partial t_2^{n_2-1} } f(t_1,t_2) |_{t_1=t_2=0}
\end{equation}

\section{Fractional diffusion equation and log-stable processes}
Under the Carr-Wu maximal negative hypothesis $\beta=-1$, or equivalently, $\theta = \alpha - 2$, it is known that the probability distributions of the log-returns (i.e., the Green functions, or fundamental solutions associated to the the option price \eqref{expect}) satisfy the equation (see \cite{Kleinert16,Kleinert13} and references therein):
\begin{equation}\label{fract_diff_carr_wu}
\frac{\partial g_\alpha (x,t)}{\partial t} \, + \, \mu \cdot [ {}^{\alpha-2} D ^{\alpha} \, g_{\alpha} (x,t)] \, = \, 0
\end{equation}
which is a particular case of the generic \textit{space-fractional diffusion}
\begin{equation}\label{fract_diff}
\frac{\partial g_\alpha (x,t)}{\partial t} \, + \, \mu \cdot [{}^{\theta} D ^{\alpha} \, g_{\alpha} (x,t)] \, = \, 0
\end{equation}
Here, the space fractional derivatives is a \textit{Riesz-Feller derivatives}, a two-parameter operator defined by its action on the Fourier space (see \cite{Gorenflo12,Mainardi05} for more details on fractional analysis):
\begin{equation}\label{def_FR}
\widehat{^\theta D ^\alpha f} (k) \, := \, |k|^\alpha e^{i \mathrm{sgn} \, k \frac{\theta \pi}{2}} \, \hat{f}(k)
\end{equation}
It is a consequence of definition \eqref{def_FR} that Green functions associated to space fractional diffusions can be written under the form
\begin{equation}
g_{\alpha,\theta} (x , t ) \, = \, g_{\alpha,\theta} \left( \frac{x}{(-\mu t)^{\frac{1}{\alpha}}} \right)
\end{equation}
These functions have been extensively studied \cite{Mainardi05}; they can conveniently be expressed as Mellin-Barnes integrals. Introducing
\begin{equation}
\rho_{\pm} \ : = \, \frac{\alpha \, \mp \, \theta}{2\alpha}
\end{equation}
then for any $X>0$ the following representation holds:
\begin{equation}\label{Mellin_Green_LS_R+}
g_{\alpha,\theta}(X) \, = \, \frac{1}{\alpha} \, \int\limits_{c_1-i\infty }^{c_1+ i \infty} \, \frac{\Gamma(\frac{t_1}{\alpha})\Gamma(1-t_1)}{\Gamma(\rho_+ t_1)\Gamma(1-\rho_+ t_1)} \, X^{t_1-1} \, \frac{\id t_1}{2i\pi} \hspace*{1cm}  0 < c_1 < \min\{\alpha,1\}
\end{equation}
The Green function \eqref{Mellin_Green_LS_R+} extends to negative arguments via the symmetry property:
\begin{equation}
g_{\alpha,\theta}(-X) \, = \, g_{\alpha,-\theta}(X)
\end{equation}
In the Finite Moment Log-Stable model, where the maximal negative asymmetry hypothesis holds, one notes that
\begin{equation}
\rho_+ \, = \, \frac{1}{\alpha} \,\,\, \mathrm{and} \,\,\, \rho_- \, = \, \frac{\alpha-1}{\alpha}
\end{equation}
and therefore we may denote the corresponding Green function by:
\begin{equation}\label{Green_FLMS_R+}
g_{\alpha}(X) \, = \, \frac{1}{\alpha} \, \int\limits_{c_1-i\infty }^{c_1+ i \infty} \,
\left[
\frac{\Gamma(1-t_1)}{\Gamma(1-\frac{t_1}{\alpha})} \mathrm{H}(X>0)
+
\frac{\Gamma(\frac{t_1}{\alpha})\Gamma(1-t_1)}{\Gamma(\frac{\alpha-1}{\alpha} t_1)\Gamma(1-\frac{\alpha-1}{\alpha} t_1)} \mathrm{H}(X<0)
\right]
\, |X|^{t_1-1} \, \frac{\id t_1}{2i\pi}
\end{equation}
where $H$ is the Heaviside step function. In particular, when $\alpha=2$ (Gaussian, i.e. Black-Scholes case), $\rho_+ = \rho_- \, = \, \frac{1}{2}$ and:
\begin{equation}\label{BS_Green}
g^{BS}(X) \, = \, \frac{1}{2} \, \int\limits_{c_1-i\infty }^{c_1+ i \infty} \, \frac{\Gamma(1-t_1)}{\Gamma(1-\frac{t_1}{2})} \, |X|^{t_1-1} \, \frac{\id t_1}{2i\pi}
\end{equation}
Using the Legendre duplication formula \cite{Abramowitz72} and performing a change of variables $t_1 \rightarrow \frac{1-t_1}{2}$ shows that \eqref{BS_Green} can be written under the form
\begin{equation}
g^{BS}(X) \, = \, \frac{1}{2\sqrt{\pi}} \, \int\limits_{c_1-i\infty }^{c_1+ i \infty} \, \Gamma(t_1) \, \left( \frac{|X|^2}{4} \right)^{-t_1} \, \frac{\id t_1}{2i\pi}
\end{equation}
Using the inversion formula \eqref{Cahen}, we get:
\begin{equation}
g^{BS}(X) \, = \, \frac{1}{2\sqrt{\pi}} e^{-\frac{|X|^2}{4}}
\end{equation}
Recalling that, in the Gaussian case, $\alpha=2$ and $\mu = -\frac{\sigma^2}{2}$, then the Green function in the convolution \eqref{Levy_Green} becomes
\begin{equation}
\frac{1}{(-\mu\tau)^{\frac{1}{\alpha}}} \, g^{BS}\left(\frac{y}{(-\mu\tau)^{\frac{1}{\alpha}}}\right) \, = \, \frac{1}{\sigma\sqrt{2\pi\tau}}\, e^{-\frac{|y|^2}{2\sigma^2\tau}}
\end{equation}
which is the well-known \textit{heat kernel}. It is the fundamental solution of the one-dimensional diffusion (heat) equation
\begin{equation}
\frac{\partial W}{\partial \tau} \, - \, \frac{\sigma^2}{2} \frac{\partial^2 W}{\partial^2 y} \, = \ 0
\end{equation}
which is the particular case of the fractional diffusion \eqref{fract_diff_carr_wu} for $\alpha = 2$, and a reduction of the Black-Scholes PDE after a suitable change of variables \cite{Wilmott06}, whose solution is given by the Black-Scholes formula:
\begin{equation}\label{B-S_formula}
V_{BS}(S,K,r,\sigma,\tau) \, = \, S N\left(\frac{[\log]}{\sigma\sqrt{\tau}}+\frac{1}{2}\sigma\sqrt{\tau} \right) \, - \, K e^{-r\tau} \, N\left(\frac{[\log]}{\sigma\sqrt{\tau}} - \frac{1}{2}\sigma\sqrt{\tau} \right)
\end{equation}
where $[\log]:=\log\frac{S}{K}+r\tau$ and $N(.)$ is the normal distribution function. A particular case occurs when the asset is "at-the-money forward", that is when
\begin{equation}
S\, = \, Ke^{-r\tau}
\end{equation}
because, then, $[\log]=0$ and the Black-Scholes formula \eqref{B-S_formula} reduces to 
\begin{equation}
V^{(ATMF)}_{BS}(S,K,r,\sigma,\tau) \, = \, S \, \left[ N(\frac{1}{2}\sigma\sqrt{\tau}) - N(-\frac{1}{2}\sigma\sqrt{\tau})  \right]
\end{equation}
In this case, it is easy to Taylor expand the normal distribution functions \cite{Abramowitz72}, with the result:
\begin{equation}\label{Brenner_Series_1}
V^{(ATMF)}_{BS}(S,K,r,\sigma,\tau) \, = \, \frac{S}{\sqrt{\pi}} \, \sum\limits_{n=0}^{\infty} \, \frac{ (-1)^n  \left( \sigma\frac{\sqrt{\tau}}{\sqrt{2}} \right)^{2n+1}}{n! 2^{2n} (2n+1)}
\end{equation}
and therefore the price goes as:
\begin{equation}
V^{(ATMF)}_{BS}(S,K,r,\sigma,\tau) \, = \, \frac{1}{\sqrt{2\pi}} S \sigma \sqrt{\tau} \, + \, O\left( (\sigma\sqrt{\tau})^3 \right) \, \simeq \, 0.4 S \sigma\sqrt{\tau}
\end{equation}
which is the approximation obtained by Brenner and Subrahmanyam in \cite{BS94}. Last, one can note that, using the particular values for the Gamma function at half integers \cite{Abramowitz72}, one can rewrite the series \eqref{Brenner_Series_1} as
\begin{equation}\label{Brenner_Series_2}
V^{(ATMF)}_{BS}(S,K,r,\sigma,\tau) \, = \, S \sum\limits_{n=0}^{\infty} \, \frac{ \left( \sigma\frac{\sqrt{\tau}}{\sqrt{2}} \right)^{2n+1}}{(2n+1)! \Gamma(\frac{1}{2}-n)}
\end{equation}

\end{document}